%% file: lin-fix.tex
\documentclass[manuscript, nonacm]{acmart}

\usepackage{cleveref}
\usepackage{tikz}

\newtheorem{claim}{Claim}

\setcopyright{none}
\copyrightyear{2021}
\acmYear{2021}
\acmDOI{}
\acmConference[]{}{}{}
\acmBooktitle{}
\acmPrice{}
\acmISBN{}
    
\begin{document}

\title{Linearizability: A Typo}
\thanks{This work was supported by the United States - Israel BSF grant No. 2018655}

\author{Gal Sela}
\affiliation{
  \institution{Technion}
  \country{Israel}
}
\email{galy@cs.technion.ac.il}

\author{Maurice Herlihy}
\affiliation{
  \institution{Brown University}
  \country{USA}
}
\email{mph@cs.brown.edu}

\author{Erez Petrank}
\affiliation{
  \institution{Technion}
  \country{Israel}
}
\email{erez@cs.technion.ac.il}

\begin{abstract}
Linearizability \cite{herlihy1990linearizability} is the de facto consistency condition for concurrent objects, widely used in theory and practice. Loosely speaking, linearizability classifies concurrent executions as correct if operations on shared objects appear to take effect instantaneously during the operation execution time. 
This paper calls attention to a somewhat-neglected aspect of linearizability: restrictions on how pending invocations are handled, an issue that has become increasingly important for software running on systems with non-volatile main memory.
Interestingly, the original published definition of linearizability includes a typo (a symbol is missing a prime) that concerns exactly this issue.  
In this paper we point out the typo and provide an amendment to make the definition complete. We believe that pointing this typo out rigorously and proposing a fix is important and timely. 
\end{abstract}



\maketitle
The conference version of this paper is available at \cite{sela2021linearizability}.

\section{Introduction}
Linearizability is the prevalent correctness condition for concurrent executions on shared objects. It determines whether a concurrent execution is correct by relating it to a sequential execution that satisfies the sequential specification of the object. 
To relate a valid sequential execution to a concurrent one, linearizability specifies an order of the concurrent operations, denoted \emph{linearization order}. On the one hand, linearizability requires that if we execute the operations sequentially one by one according to their linearization order (with the same parameters as in the concurrent execution), we obtain a sequential execution (with the same operation results as in the concurrent execution) that satisfies the sequential specification of the object. On the other hand, linearizability dictates that the linearization order preserve the order of non-overlapping operations in the original concurrent execution.
Namely, if an operation {\em op$_1$} completes before another operation {\em op$_2$} begins, then {\em op$_1$} must precede {\em op$_2$} in the linearization order.
A concurrent execution is called linearizable if it can be related as above to a legal sequential execution (i.e., one that satisfies the sequential specification of the object). The formal definition is provided in~\Cref{orig-lin-def-section}. 

So far, we ignored pending invocations in the execution. These are invocations of operations that start during the execution, but do not complete. The issue that we point out in this paper concerns the treatment of pending invocations. The original linearizability definition provides a treatment of pending invocations, stating which executions with pending invocations are linearizable. However, in this paper we argue that, due to a typo, this treatment of pending invocations is lacking. It contradicts our intuition about linearizable executions, and furthermore, causes the locality and nonblocking properties to not hold. We then propose a simple fix for the typo that fits intuition and obtains these desirable properties also for executions with pending invocations.  

Linearizability allows eliding some of the pending invocations, namely, excluding their operations from the related sequential execution. This can be interpreted as operations that do not yet take effect before the execution ends. The rest of the pending invocations appear in the sequential execution with a {\em response}, i.e., completion and returned results. 
These can be interpreted as operations in the concurrent execution that have taken effect although their responses have not yet been returned to the caller.
The responses appended in the sequential execution are determined in a way that fits the overall execution. In particular, responses are set so that the related sequential execution satisfies the sequential specification of the object. 

The problem that arises, due to the typo in the original definition, is that the definition does not restrict the order of operations that have pending invocations, even when these operations take effect, i.e., are included in the related sequential execution. In particular, a pending invocation of an operation {\em op} may be placed in a linearization order before other operations that completed earlier in the execution, even operations that completed before {\em op} started. Imagine an execution that starts with an operation {\em op$_1$} that reads a shared variable $x$. While $x$ is initially $0$, the operation weirdly reads the value $1$, and then completes and returns $1$. Later a new operation {\em op$_2$} is invoked on a different process. Operation {\em op$_2$} writes $1$ into $x$ and does not complete before the execution ends. Intuitively, this does not seem like an acceptable linearizable execution. However, under the existing definition with the typo, it is linearizable, because the pending invocation of operation {\em op$_2$} (that writes $1$) can be ordered before the completed operation {\em op$_1$} that reads $1$. 

Interestingly, beyond contradicting intuition, the typo in the original definition does not allow it to yield neither locality nor the nonblocking property. 
In \Cref{section:issues} we describe the intuitive problem with the typo in the definition and formally show that it is not local neither nonblocking.  

We propose a (syntactically very minor) modification to the definition that restricts the linearization order of operations with pending invocations that take effect. 
Similarly to completed operations, operations with pending invocations are ordered later than any operation that completes before they start. This modification makes the odd execution described above not linearizable. In \Cref{sec:original} we recall the formal original definition of linearizability (with the typo). In \Cref{sec:ammended-linearizability} we formally specify the amended definition, and in \Cref{sec-revisit} we revisit the issues that the typo raises and show that they are resolved. 
We also show in \Cref{equivalent-def-section} that an alternative equivalent definition of linearizability \cite{lynch1996distributed} is actually not equivalent to the definition of linearizability with the typo, but it is equivalent to the version that we propose without the typo. 
In \Cref{alternative-def-section} we discuss an alternative interpretation of the original definition (with the typo, but with a different definition of what an operation means) and explain why this alternative interpretation is problematic as well.  
Finally, we put the various linearizability definitions covered throughout the paper in context in \Cref{def-comparison-section}.

It is clear that the flaw in the definition is a typo, not a conceptual error, and the authors' intended meaning is clear in context. We believe no prior paper was rendered incorrect by relying on the original definition. However, linearizability is extremely important for concurrent executions. It has been used in thousands of papers and the definition with the typo has been replicated in numerous subsequent publications. We believe it is important to point out this typo and provide a rigorous discussion and a fix. 
The issue of pending invocations is becoming increasingly important as architectures with non-volatile main memory become commonplace. Non-volatile memory models encompass various definitions \cite[e.g.][]{izraelevitz2016linearizability,aguilera2003strict,guerraoui2004robust,berryhill2015robust} where a major focus is dealing with invocations pending at the time of a crash.
In this realm, pending invocations become critically important, making the fix of this typo timely.

\section{System Model and Linearizability Definition}
\label{sec:original}
We follow the terminology of the original linearizability paper~\cite{herlihy1990linearizability}, which contains additional motivation and detailed explanations.  
\subsection{Histories Terminology}
An execution of a concurrent system is modeled by a history.
A \emph{history} is a finite sequence of operation \emph{invocation} and \emph{response} events.
Each invocation or response event is associated with some object and some process. An invocation includes also an operation name and argument values, and a response includes a termination condition and results.
A response \emph{matches} an invocation if it is associated with the same object and process. An invocation is \emph{pending} in a history if no matching response follows the invocation. An \emph{extension} of $H$ is a history constructed by appending to the end of $H$ responses to zero or more pending invocations of $H$. 
A \emph{subhistory} of a history $H$ is a subsequence of the events of $H$.
\emph{complete(H)} is the maximal subhistory of $H$ consisting only of invocations and matching responses, without any pending invocations.
For a process $P$, the \emph{process subhistory} $H|P$ is the subsequence of all events in $H$ associated with the process $P$. 
For an object $x$, the \emph{object subhistory} $H|x$ is the subsequence of all events in $H$ associated with the object $x$.
Two histories $H$ and $H'$ are \emph{equivalent} if for every process $P$, $H|P=H'|P$.

A history $H$ is \emph{sequential} if it comprises a sequence of pairs of an invocation and a matching response, except possibly the last invocation, which might be the last event in the history, not accompanied by a matching response. A history that is not sequential is \emph{concurrent}.
A history is \emph{well-formed} if each of its process subhistories is sequential.
A \emph{single-object} history is one in which all events are associated with the same object.
A \emph{sequential specification} for an object is a prefix-closed set of single-object sequential histories for that object. 
A sequential history $H$ is \emph{legal} if each object subhistory $H|x$ belongs to the sequential specification for $x$.

We recall the exact definition of operations from \cite{herlihy1990linearizability}, which is inherent to the definition of linearizability:
\begin{definition}\label{op-def}
\emph{(Operation)}
An \emph{operation} in a history is a pair consisting of an invocation and the next matching response. 
\end{definition}
An operation $e_0$ \emph{precedes} (synonymously \emph{happens before}) an operation $e_1$ in a history $H$ if $e_0$ ends before $e_1$ begins, namely, $e_1$'s invocation event occurs after $e_0$'s response event in $H$.
Precedence in $H$ induces a partial order on operations of $H$, denoted $<_H$.
Informally, $<_H$ captures the "real-time" precedence order of operations in $H$.
We stress that only invocations that have matching responses are considered \emph{operations} and the order $<_H$ applies only to them.

\subsection{Linearizability Definition}\label{orig-lin-def-section}
The original definition of linearizability according to \cite{herlihy1990linearizability} follows:
\begin{definition}\label{orig-def}
\emph{(Linearizability)}
A well-formed history $H$ is \emph{linearizable} if it has an extension $H'$ such that:
\begin{enumerate}
\item[\textbf{L1:}]\vspace{-0.1cm}
There exists a legal sequential history $S$, to which $complete(H')$ is equivalent.
\item[\textbf{L2:}]
$<_H\subseteq<_S$.
\end{enumerate}
\end{definition}
$S$ is denoted the \emph{linearization} of $H$, and operations that appear in {\em complete}$(H')$ are denoted {\em linearized} operations. 
\Cref{orig-def} requires the existence of a linearization $S$ that satisfies two conditions. 
Condition L1 refers to each process individually, guaranteeing that all its linearized operations are in the same order and with the same results as in the legal sequential history $S$. Due to this equivalence to a legal sequential history, operations in $H$ act as if they were interleaved at the granularity of complete operations, and adhere to the sequential specification. 
Condition L2 guarantees that $S$ preserves the order of non-concurrent operations in $H$, so that it respects possible dependencies between operations in $H$.

\section{Issues with the definition with the typo}
\label{section:issues}
Linearizability enforces real-time precedence order on operations. \Cref{orig-def} enforces it only on operations that include both an invocation and a response in the given history. Fixing the typo extends the enforcement to linearized operations related to pending invocations as well, so that overall, the order is enforced on all linearized operations.
To establish the necessity of the typo fix, we present in \Cref{examples-subsection} motivating examples of executions classified as linearizable by \Cref{orig-def} although intuitively they do not seem like acceptable linearizable executions. Moreover, we show in \Cref{locality-subsection} that linearizability as defined in \Cref{orig-def} is not local, and show in \Cref{nonblocking-subsection} that it is not nonblocking.

\subsection{Executions Counter-Intuitively Classified As Linearizable} \label{examples-subsection}
We start with two simple examples of executions on a single object, that intuitively seem non-linearizable, but are classified as linearizable by \Cref{orig-def}. This stems from allowing operations related to pending invocations to appear to take effect before operations by other processes that precede them, since L2 enforces order among $H$'s operations only, and excludes all pending invocations of $H$. 

Consider the execution $H_1$ that appears in \Cref{fig:register-histories}, involving two processes: $A$ and $B$, operating on a register object initialized to 0.
$H_1$ is intuitively unacceptable, as $A$ cannot "predict the future" and read the value that $B$ has not yet even asked to write, and it cannot distinguish between the given execution and an execution in which $B$ does not invoke any write. Therefore, $A$ should return 0 and not 1.
Nevertheless, \Cref{orig-def} classifies the execution as linearizable, as there is an extension $H_1'$ and a legal sequential execution $S_1$ (see \Cref{fig:register-histories}) that adhere to the conditions in \Cref{orig-def}:
L1 holds since the events per process in $complete(H_1')$ and in $S_1$ are identical.
L2 vacuously holds since it enforces order only on operations of $H_1$, and $H_1$ has a single operation (since a pending invocation does not count as an operation, see \Cref{op-def}). In particular, L2 does not force $B$'s write to occur in $S_1$ after the read operation by $A$.

\setlength{\unitlength}{0.1mm}
\begin{figure}[b]

\captionsetup{singlelinecheck=off}
\caption[]{
\label{fig:register-histories}
Executions on a register:
\begin{itemize}
    \item 
      $H_1$ -- an execution with a return value conflicting with the order between the pending write and the preceding read.
    \item
      $complete(H_1')$ -- identical to $H_1'$, an extension of $H_1$.
    \item 
      $S_1$ -- a linearization of $H_1$.
\end{itemize}
}

\begin{picture}(800,160)
\put(10,45){$H_1$:}

\put(260,90){A}
\put(260,0){B}

\put(320,115){\line(0,-1){30}} 
\put(320,100){\line(1,0){200}} 
\put(520,115){\line(0,-1){30}} 
\put(375,115){Read()}
\put(512,130){1} 

\put(570,25){\line(0,-1){30}} 
\put(570,10){\line(1,0){200}} 
\put(610,25){Write(1)}
\multiput(790,10)(20,0){3}{\line(1,0){5}}

\end{picture}

\vspace{10pt}\par

\begin{picture}(800,160)
\put(0,160){\textcolor{gray}{\line(1,0){850}}}

\put(10,45){$complete(H_1')$:}

\put(260,90){A}
\put(260,0){B}

\put(320,115){\line(0,-1){30}} 
\put(320,100){\line(1,0){200}} 
\put(520,115){\line(0,-1){30}} 
\put(375,115){Read()}
\put(512,130){1} 

\put(570,25){\line(0,-1){30}} 
\put(570,10){\line(1,0){200}} 
\put(610,25){Write(1)}
\put(770,25){\line(0,-1){30}} 
\end{picture}

\vspace{10pt}\par

\begin{picture}(800,160)
\put(0,160){\textcolor{gray}{\line(1,0){850}}}

\put(10,45){$S_1$:}

\put(260,90){A}
\put(260,0){B}

\put(570,115){\line(0,-1){30}} 
\put(570,100){\line(1,0){200}} 
\put(770,115){\line(0,-1){30}} 
\put(625,115){Read()}
\put(762,130){1} 

\put(320,25){\line(0,-1){30}} 
\put(320,10){\line(1,0){200}} 
\put(520,25){\line(0,-1){30}} 
\put(360,25){Write(1)}
\end{picture}

\end{figure}

We bring as a second example the execution $H_2$ demonstrated in \Cref{fig:queue-histories}, involving two processes: $A$ and $B$, operating on a FIFO queue initialized to be empty.
$H_2$ is intuitively unacceptable due to the return value of the dequeue operation, which returns an item enqueued by the second enqueue operation rather than the first one, thus violating the FIFO requirement.
Nonetheless, as Condition L2 of \Cref{orig-def} does not enforce linearizations of $H_2$ to place $A$'s enqueue before $B$'s enqueue, $H_2$ is considered linearizable, by the extension $H_2'$ and the linearization $S_2$ that appear in \Cref{fig:queue-histories}.

\setlength{\unitlength}{0.1mm}
\begin{figure}
\captionsetup{singlelinecheck=off}
\caption[]{
\label{fig:queue-histories}
Executions on a FIFO queue:
\begin{itemize}
    \item 
      $H_2$ -- an execution with a return value conflicting with the order between the pending enqueue and the preceding enqueue.
    \item
      $complete(H_2')$ -- identical to $H_2'$, an extension of $H_2$.
    \item 
      $S_2$ -- a linearization of $H_2$.
\end{itemize}
}
\begin{picture}(750,160)

\put(-230,45){$H_2$:}
\put(20,90){A}
\put(20,0){B}

\put(100,115){\line(0,-1){30}} 
\put(100,100){\line(1,0){200}} 
\put(300,115){\line(0,-1){30}} 
\put(155,115){Enq(x)}

\put(400,115){\line(0,-1){30}} 
\put(400,100){\line(1,0){200}} 
\put(600,115){\line(0,-1){30}} 
\put(455,115){Deq()}
\put(592,130){y} 

\put(350,25){\line(0,-1){30}} 
\put(350,10){\line(1,0){300}} 
\put(448,25){Enq(y)}
\multiput(670,10)(20,0){3}{\line(1,0){5}}
\end{picture}

\vspace{10pt}\par

\begin{picture}(750,160)
\put(-240,160){\textcolor{gray}{\line(1,0){1210}}}

\put(-230,45){$complete(H_2')$:}
\put(20,90){A}
\put(20,0){B}

\put(100,115){\line(0,-1){30}} 
\put(100,100){\line(1,0){200}} 
\put(300,115){\line(0,-1){30}} 
\put(155,115){Enq(x)}

\put(400,115){\line(0,-1){30}} 
\put(400,100){\line(1,0){200}} 
\put(600,115){\line(0,-1){30}} 
\put(455,115){Deq()}
\put(592,130){y} 

\put(350,25){\line(0,-1){30}} 
\put(350,10){\line(1,0){300}} 
\put(448,25){Enq(y)}
\put(650,25){\line(0,-1){30}} 
\end{picture}

\vspace{10pt}\par

\begin{picture}(750,160)
\put(-240,160){\textcolor{gray}{\line(1,0){1210}}}

\put(-230,45){$S_2$:}
\put(20,90){A}
\put(20,0){B}

\put(450,115){\line(0,-1){30}} 
\put(450,100){\line(1,0){200}} 
\put(650,115){\line(0,-1){30}} 
\put(505,115){Enq(x)}

\put(750,115){\line(0,-1){30}} 
\put(750,100){\line(1,0){200}} 
\put(950,115){\line(0,-1){30}} 
\put(805,115){Deq()}
\put(942,130){y} 

\put(100,25){\line(0,-1){30}} 
\put(100,10){\line(1,0){300}} 
\put(198,25){Enq(y)}
\put(400,25){\line(0,-1){30}} 
\end{picture}

\end{figure}

\subsection{Linearizability With The Typo Is Not Local}
\label{locality-subsection}

A property of a concurrent system is \emph{local} (synonymously \emph{composable}) if the system as a whole satisfies it whenever each object in the system satisfies it individually.
Locality enables implementing and verifying objects independently, thus maintaining modularity. 
To demonstrate that linearizability as defined with the typo is not local,
we bring an execution $H$ in \Cref{fig:2-registers}, involving two processes: A and B, operating on two register objects initialized to 0: x and y. 
For each of x and y, the object subhistory of the presented execution $H$ is similar to $H_1$ (see \Cref{fig:register-histories}). As shown in \Cref{examples-subsection}, these subhistories are linearizable by \Cref{orig-def}. However, $H$ as a whole is not linearizable by \Cref{orig-def}: 
An appropriate extension $H'$ must include responses to both writes for the writes to be included in $complete(H')$, 
otherwise there will be no legal sequential execution $S$ equivalent to $complete(H')$, because the read operations could not legally return 1. 
Together with the order enforced by Condition L1 on operations by each process, we get the following order requirements, which form a cycle:
x.Read() must occur before y.Write(1) for $S$ to preserve the order of A's events (due to Condition L1 that requires $S|A=complete(H')|A$);
y.Write(1) must occur before y.Read() for $S$ to be a legal register history (which dictates in particular that 1 be a legal return value of y.Read());
y.Read() must occur before x.Write(1) for $S$ to preserve the order of B's events (due to Condition L1 that requires $S|B=complete(H')|B$); 
and finally x.Write(1) must occur before x.Read() for $S$ to be a legal register history.

\setlength{\unitlength}{0.1mm}
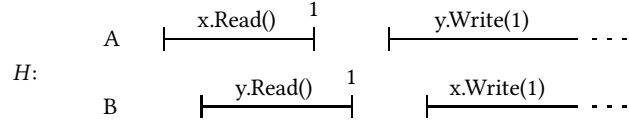
\begin{figure}[t]
\caption{$H$, a non-linearizable execution on two registers, although the object subhistory for each register is linearizable}
\label{fig:2-registers}
\begin{picture}(700,160)
\put(-100,45){$H$:}
\put(20,90){A}
\put(20,0){B}

\put(100,115){\line(0,-1){30}} 
\put(100,100){\line(1,0){200}} 
\put(300,115){\line(0,-1){30}} 
\put(145,115){x.Read()}
\put(292,130){1} 

\put(150,25){\line(0,-1){30}} 
\put(150,10){\line(1,0){200}} 
\put(350,25){\line(0,-1){30}} 
\put(195,25){y.Read()}
\put(342,40){1} 

\put(400,115){\line(0,-1){30}} 
\put(400,100){\line(1,0){250}} 
\put(460,115){y.Write(1)}
\multiput(670,100)(20,0){3}{\line(1,0){5}}

\put(450,25){\line(0,-1){30}} 
\put(450,10){\line(1,0){200}} 
\put(480,25){x.Write(1)}
\multiput(670,10)(20,0){3}{\line(1,0){5}}

\end{picture}

\end{figure}

\subsection{Linearizability With The Typo Is Not Nonblocking}
\label{nonblocking-subsection}
We look at pending invocations of {\em total} operations, which are operations defined for every object value (following the terminology of \cite{herlihy1990linearizability}).
A property of a concurrent system is \emph{nonblocking} if processes invoking total operations are never forced to wait for another pending invocation to complete.
Formally, linearizability is nonblocking if for each linearizable execution that has some pending invocation $inv$ of a total operation, there exists a matching response for $inv$ such that appending it to the execution results in a linearizable execution.

To demonstrate that linearizability as defined with the typo is not nonblocking, we look at the execution $H_1$ (see \Cref{fig:register-histories}).
This execution is linearizable by the original definition (as shown above) and has a pending invocation -- the invocation of B's write. Appending a response $resp$ for this write to $H_1$ results in the non-linearizable execution $H_1\cdot resp$: The pending write becomes an operation in $H_1\cdot resp$ and is thus ordered by $<_{H_1\cdot resp}$ as appearing after A's read. Condition L2 of \Cref{orig-def} dictates that a linearization of $H_1\cdot resp$ respect this order and place the write after the read, which makes it impossible for A's read to legally return 1.

Another counterexample is the execution $H_2$ (see \Cref{fig:queue-histories}). It is linearizable by the original definition (as shown above) and has a pending invocation -- the invocation of B's enqueue.
Appending a response $resp$ for this enqueue to $H_2$ results in the non-linearizable execution $H_2\cdot resp$: The pending enqueue becomes an operation in $H_2\cdot resp$ and is thus ordered by $<_{H_2\cdot resp}$ as appearing after A's enqueue. Condition L2 of \Cref{orig-def} dictates that the linearization respect this order and place B's enqueue after A's enqueue, which makes it impossible for A's dequeue to return y while obeying the FIFO specification of a FIFO queue.

\section{Amended Linearizability}
\label{sec:ammended-linearizability}
We bring the amended definition of linearizability, which fixes a typo in Condition L2
with a fix that enforces real-time precedence order on linearized operations related to pending invocations:
\begin{definition}\label{amended-def}
\emph{(Amended Linearizability)}
A well-formed history $H$ is \emph{linearizable} if it has an extension $H'$ such that:
\begin{enumerate}
\item[\textbf{L1:}]\vspace{-0.1cm}
There exists a legal sequential history $S$, to which $complete(H')$ is equivalent.
\item[\textbf{L2:}]
$<_{complete(H')}\subseteq<_S$.
\end{enumerate}
\end{definition}
L2 is equivalent to $<_{H'}\subseteq<_S$ because $complete(H')$ and $H'$ differ only in pending invocations and according to the definitions in~\cite{herlihy1990linearizability}, which we use too, pending invocations are not considered operations and a happens-before order does not apply to them. 
While writing L2 as above makes the definition easier to understand, note that writing L2 as $<_{H'}\subseteq<_S$ provides a fix that is within a single prime sign from the original definition. This missing prime is the typo in the original definition.

Some papers have used an alternative definition of operations, in which pending invocations are also considered as operations, to which a happens-before relation applies. We consider this alternative definition (with the original linearizability definition) in \Cref{sec:alternative} and show that it does not yield an adequate definition for linearizability. 

\section{Issues Revisited}
\label{sec-revisit}
We explain how the typo fix solves the issues raised in \Cref{section:issues}.

\subsection{Executions Become Non-Linearizable As Expected}
\label{examples-revisited}
We have shown in \Cref{examples-subsection} that $H_1$ (see \Cref{fig:register-histories}) is linearizable by \Cref{orig-def} although intuitively it is not an acceptable linearizable execution. With the amendment of L2, $H_1$ becomes non-linearizable by \Cref{amended-def}: To satisfy Condition L1, an appropriate extension $H_1'$ must include a response for the write operation, otherwise there will be no legal sequential execution $S_1$ equivalent to $complete(H_1')$, because the read operation could not legally return 1. In addition, an appropriate linearization $S_1$ must satisfy the amended L2 Condition, which dictates that the read precede the write in $S_1$ because it precedes it in $complete(H_1')$. This means the read must return 0 in $S_1$ for $S_1$ to be legal, which contradicts Condition L1 that requires $S_1$ to have the same return values as $complete(H_1')$.

In a similar fashion, $H_2$ (see \Cref{fig:queue-histories}) becomes non-linearizable by \Cref{amended-def} as desired: an appropriate extension $H_2'$ must include a response for Enq(y); thus, Enq(y) must happen in $S$ after Enq(x); and so Deq() must return $x$ in $S$, in contradiction to L1.

\subsection{Linearizability Becomes Local}
With the typo fix, the history $H$ demonstrated in \Cref{fig:2-registers} does not stand anymore as a counterexample to the locality of linearizability, since $H$'s object subhistories for $x$ and $y$ are not linearizable by \Cref{amended-def}, as explained for $H_1$ in \Cref{examples-revisited}. We will show that linearizability with the typo fix is indeed local.

Next, we repeat the locality proof from the original paper, explain why the proof does not hold for the definition with the typo, and describe an amendment of the proof that makes it correct for \Cref{amended-def}. We start by recalling the theorem and its proof as in the original paper. (We note that $H$ mentioned in the theorem is assumed to be a well-formed history, as all histories in \cite{herlihy1990linearizability}.)
 
\begin{quote} 
 
\begin{theorem}\label{locality-theorem}
$H$ is linearizable if and only if, for each object $x$, $H|x$ is linearizable.
\end{theorem} 
\begin{proof}\label{locality-proof}
The "only if" part is obvious.

For each $x$, pick a linearization of $H|x$. Let $R_x$ be the set of responses appended to $H|x$ to construct that linearization, and let $<_x$ be the corresponding linearization order. Let $H'$ be the history constructed by appending to $H$ each response in $R_x$. We will construct a partial order $<$ on the operations of $complete(H')$ such that: (1) For each $x$, $<_x \subseteq <$, and (2) $<_H \subseteq <$. Let $S$ be the sequential history constructed by ordering the operations of $complete(H')$ in any total order that extends $<$. Condition (1) implies that $S$ is legal, hence that L1 is satisfied, and Condition (2) implies that L2 is satisfied.

Let $<$ be the transitive closure of the union of all $<_x$ with $<_H$. It is immediate from the construction that $<$ satisfies Conditions (1) and (2), but it remains to be shown that $<$ is a partial order. We argue by contradiction.
If not, then there exists a set of operations $e_1, \dots, e_n$, such that $e_1<e_2<\dots<e_n$, $e_n<e_1$, and each pair is directly related by some $<_x$ or by $<_H$. Choose a cycle whose length is minimal.

Suppose all operations are associated with the same object $x$. Since $<_x$ is a total order, there must exist two operations $e_{i-1}$ and $e_i$ such that $e_{i-1}<_H e_i$ and $e_i<_x e_{i-1}$, contradicting
the linearizability of $x$.

The cycle must therefore include operations of at least two objects. By reindexing if necessary, let $e_1$ and $e_2$ be operations of distinct objects. Let $x$ be the object associated with $e_1$. We claim that none of $e_2, \dots, e_n$ can be an operation of $x$. The claim holds for $e_2$ by construction.
Let $e_i$ be the first operation in $e_3, \dots, e_n$, associated with $x$. Since $e_{i-1}$ and $e_i$ are unrelated by $<_x$, they must be related by $<_H$; hence the response of $e_{i-1}$ precedes the invocation
of $e_i$. The invocation of $e_2$ precedes the response of $e_{i-1}$, since otherwise $e_{i-1} <_H e_2$, yielding the shorter cycle $e_2,\dots, e_{i-1}$. Finally, the response of $e_1$ precedes the invocation of $e_2$, since $e_1 <_H e_2$ by construction.
It follows that the response to $e_1$ precedes the invocation of $e_i$, hence $e_1 <_H e_i$, yielding the shorter cycle $e_1, e_i, \dots, e_n$.

Since $e_n$ is not an operation of $x$, but $e_n < e_1$, it follows that $e_n <_H e_1$. But $e_1 <_H e_2$ by construction, and because $<_H$ is transitive, $e_n <_H e_2$, yielding the shorter cycle $e_2, \dots, e_n$, the final contradiction. 

\end{proof}
\end{quote}

The above proof does not hold for the definition with the typo since the history $S$ constructed in the proof is not guaranteed to satisfy Condition L1 of linearizability. 
L1 requires $S$ to preserve the precedence order of the linearized operations by each process.
But the constructed $S$ is not guaranteed to preserve the order between two linearized operations by the same process on two different objects, if the later operation of the two is related to a pending invocation in $H$. For operations on different objects, $S$ is guaranteed to preserve order only among linearized operations whose invocations are not pending in $H$, due to Condition (2) with the typo in the proof (namely, $<_H \subseteq <$ rather than $<_{complete(H')} \subseteq <$). 
For example, the history $S$ demonstrated in \Cref{fig:2-registers-wrong-lin} may be constructed by the proof (when using \Cref{orig-def}) as a linearization of $H$ from \Cref{fig:2-registers}. Due to the reversed order for $B$, $complete(H')|B\neq S|B$, and so $S$ does not satisfy L1 for $H$.

\setlength{\unitlength}{0.1mm}
\begin{figure}[b]
\caption{$S$, a sequential history that might be constructed in the proof of \Cref{locality-theorem} as a linearization of $H$}
\label{fig:2-registers-wrong-lin}
\begin{picture}(1050,160)
\put(-100,45){$S$:}
\put(20,90){A}
\put(20,0){B}

\put(350,115){\line(0,-1){30}} 
\put(350,100){\line(1,0){200}} 
\put(550,115){\line(0,-1){30}} 
\put(395,115){x.Read()}
\put(542,130){1} 

\put(950,25){\line(0,-1){30}} 
\put(950,10){\line(1,0){200}} 
\put(1150,25){\line(0,-1){30}} 
\put(995,25){y.Read()}
\put(1142,40){1} 

\put(650,115){\line(0,-1){30}} 
\put(650,100){\line(1,0){250}} 
\put(900,115){\line(0,-1){30}} 
\put(710,115){y.Write(1)}

\put(100,25){\line(0,-1){30}} 
\put(100,10){\line(1,0){200}} 
\put(300,25){\line(0,-1){30}} 
\put(130,25){x.Write(1)}

\end{picture}

\end{figure}
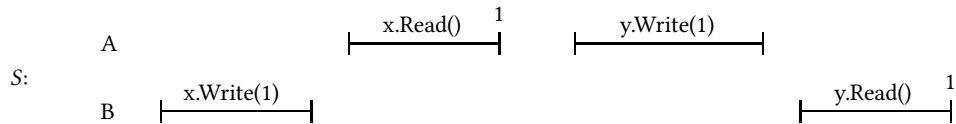

To fix the proof, we replace every reference to $<_H$ in the proof with $<_{complete(H')}$, similarly to the typo fix of the definition. 
The fixed proof is correct when applied to the linearizability definition with the typo fix. In particular, thanks to the fix of Condition (2), L1 is guaranteed to be satisfied.
We note that the fixed proof naturally holds for the definition with the typo fixed, but it does not work for the linearizability definition with the typo, which is not local as proven above. The barrier in this case is that the constructed $<$ is not necessarily a partial order, but rather might contain cycles. In detail, the above proof argues by contradiction that $<$ is a partial order. It assumes an operation cycle $e_1<e_2<\dots<e_n<e_n<e_1$ satisfying certain properties and needs to reach a contradiction. In the case that all operations in this cycle are associated with the same object $x$, then as the proof states, there must exist two operations in the cycle, $e_{i-1}$ and $e_i$, such that $e_{i-1}<_{complete(H')} e_i$ (referring to $<_{complete(H')}$ rather than $<_H$ is due to the proof fix) and $e_i<_x e_{i-1}$. This contradicts the linearizability of $H|x$ if linearizability is defined with the typo fix, namely, L2 for $H|x$ requires $<_{complete(H')|x}\subseteq<_x$, because then $e_{i-1}<_{complete(H')} e_i$ leads to $e_{i-1}<_x e_i$, contradicting the asymmetry of $<_x$. However, without the definition fix (namely, with L2 for $H|x$ requiring $<_{H|x}\subseteq<_x$), we do not reach a contradiction if $e_i$ is an operation related to a pending invocation of $H$. The reason is that for such an $e_i$, $e_{i-1}<_{complete(H')} e_i$ does not lead to $e_{i-1}<_{H|x} e_i$ (since $<_{H|x}$ does not refer to $e_i$ at all), and so L2 does not imply $e_{i-1}<_x e_i$.
Consider, for example, $H$ from \Cref{fig:2-registers}, and let $e_1$ be x.Read() and $e_2$ be x.Write(1). $<$ constructed in the proof contains the cycle $e_1<e_2<e_1$, which stems from $e_1<_{complete(H')} e_2$ and $e_2<_x e_1$.

\subsection{Linearizability Becomes Nonblocking}
With the typo fix, the histories $H_1$ and $H_2$ do not stand anymore as counterexamples to linearizability being nonblocking, since they are not linearizable by \Cref{amended-def}, as explained in \Cref{examples-revisited}. We will show that linearizability with the typo fix is indeed nonblocking.

Next, we repeat the nonblocking property proof from the original paper, explain why the proof does not hold for the definition with the typo, and why it does hold for \Cref{amended-def}. We start by recalling the theorem and its proof as in the original paper. 
 
\begin{quote} 
 
\begin{theorem}\label{nonblocking-theorem}
Let $inv$ be an invocation of a total operation. If $\langle x\ inv\ P\rangle$ is a pending invocation in a linearizable history $H$, then there exists a response $\langle x\ res\ P\rangle$ such that $H\cdot\langle x\ res\ P\rangle$ is linearizable. 
\end{theorem} 
\begin{proof}\label{nonblocking-proof}
Let $S$ be any linearization of $H$. If $S$ includes a response $\langle x\ res\ P\rangle$ to $\langle x\ inv\ P\rangle$, we are done, since $S$ is also a linearization of $H\cdot\langle x\ res\ P\rangle$. 
Otherwise, $\langle x\ inv\ P\rangle$ does not appear in $S$ either, since linearizations, by definition, include no pending invocations. Because the operation is total, there exists a response $\langle x\ res\ P\rangle$ such that $S' = S \cdot \langle x\ inv\ P\rangle \cdot \langle x\ res\ P\rangle$ is legal. $S'$, however, is a linearization of $H\cdot \langle x\ res\ P\rangle$, and hence is also a linearization of $H$. 
\end{proof}

\end{quote}

The first part of the above proof does not hold for the definition with the typo, as Condition L2 of \Cref{orig-def}, applied to the history $H\cdot\langle x\ res\ P\rangle$, might not be satisfied by $S$:
$S$, as a linearization of $H$, is guaranteed by L2 of \Cref{orig-def} to respect the precedence order $<_H$. The issue is that $<_H$ does not enforce any order on pending invocations of $H$, in particular $\langle x\ inv\ P\rangle$. Since this invocation is not pending in the history $H\cdot\langle x\ res\ P\rangle$, L2 for this history requires more than L2 for $H$ guarantees. In particular, it requires that a linearization of $H\cdot\langle x\ res\ P\rangle$ order the operation of $\langle x\ inv\ P\rangle$ after any operation completed beforehand. However, $S$ might not satisfy it (like in the examples brought in \Cref{nonblocking-subsection}).

The first part of the proof does hold for the amended definition:
If $S$ is a linearization of $H$ by \Cref{amended-def} that includes a response $\langle x\ res\ P\rangle$ to $\langle x\ inv\ P\rangle$, then there exists an extension $H'$ of $H$, which includes $\langle x\ res\ P\rangle$ as the first response appended after $H$, such that $S$ and $H'$ satisfy \Cref{amended-def} for $H$. 
This $H'$ is also an extension of $H\cdot\langle x\ res\ P\rangle$, and so the same $S$ and $H'$ satisfy Conditions L1 and L2 of \Cref{amended-def} for $H\cdot\langle x\ res\ P\rangle$ as well.

\input{alternative-def}

\input{equivalent-def}

\section{Comparison of all definition versions}
\label{def-comparison-section}
Let $H$ be a history. Then for each choice of an extension $H'$, we may divide the invocations in $H$ into three categories: 
\begin{enumerate}
\item\label{non-pending-invocations}
Invocations that have a matching response in $H$. 
\item\label{pending-linearized-invocations}
Pending invocations in $H$ that have a matching response in $H'$.
\item\label{pending-non-linearized-invocations}
Pending invocations in $H$ that do not have a matching response in $H'$.
\end{enumerate}
The operations related to invocations of categories (\ref{non-pending-invocations}) and (\ref{pending-linearized-invocations})
form $complete(H')$. Therefore, as implied by Condition L1 of linearizability, if $H$ is linearizable then the linearization of $H$ which stems from this $H'$ is made of these operations.

The different versions of the linearizability definition differ in the operations among which their L2 Condition enforces order preservation.
\Cref{orig-def} forces preserving order among operations whose invocations are of category (\ref{non-pending-invocations}) only.
\Cref{amended-def} (and thus the equivalent \Cref{equivalent-def} as well) dictates that the linearization preserve the order of all operations in $complete(H')$, namely, all linearized operations, which are operations whose invocations are of categories (\ref{non-pending-invocations}) and (\ref{pending-linearized-invocations}).
Since L2 with the typo fix enforces precedence order on additional operations in comparison to the original L2, then the typo fix strengthens the definition of linearizability. Namely, the amended definition eliminates some executions that are linearizable according to the definition with the typo, e.g., executions $H_1$ (see \Cref{fig:register-histories}) and $H_2$ (see \Cref{fig:queue-histories}). 
The typo fix does not include any execution not classified as linearizable by the original definition: each execution linearizable by \Cref{amended-def} is linearizable by \Cref{orig-def} as well because $<_H\subseteq<_{complete(H')}$.

\begin{figure}[b]
\caption{The relationship between histories categorized as linearizable by the different versions of linearizability}
\label{fig:venn}
\begin{tikzpicture}
    \draw[purple] (-1,0) ellipse (6.6cm and 2.5cm);
    \draw[teal] (0.6,0) ellipse (4.5cm and 2.1cm);
    \draw[violet] (2,0) ellipse (2.2cm and 1.6cm);
    
    \node[purple] at (-5.9,0.3) {Histories linearizable};
    \node[purple] at (-5.9,-0.3) {by \Cref{orig-def}};
    \node[teal] at (-2.2,0.3) {Histories linearizable};
    \node[teal] at (-2.2,-0.3) {by \Cref{amended-def}};
    \node[violet] at (2,0.6) {Histories linearizable by};
    \node[violet] at (2,0) {the alternative interpretation};
    \node[violet] at (2,-0.6) {of \Cref{orig-def}};
    
    \node at (-5.5,-1) {\LARGE\textbf{.}};
    \node at (-5.5,-1.3) {$H_1$};
    \node at (-4.5,-1) {\LARGE\textbf{.}};
    \node at (-4.5,-1.3) {$H_2$};
    \node at (-1.6,-1) {\LARGE\textbf{.}};
    \node at (-1.6,-1.3) {$H_s$};
\end{tikzpicture}
\end{figure}

\Cref{orig-def} with the alternative interpretation of an operation as described in \Cref{sec:alternative}, forces preserving order among operations whose invocations are of any of the 3 above-mentioned categories, namely, operations related to all invocations. 
L2 in this interpretation implies that a linearization $S$ includes all pending invocations of $H$ (that are preceded by any response), including those of category (\ref{pending-non-linearized-invocations}) -- which cannot appear in $S$ by Condition L1. Consequently, for a history to be classified as linearizable, there must exist an extension $H'$ for which $H$ has no invocations of category (\ref{pending-non-linearized-invocations}), namely, all pending invocations are linearized (again, referring only to invocations preceded by a response).
Hence, the definition in this interpretation excludes histories with pending invocations (preceded by some response) that cannot be linearized, like $H_s$ (see \Cref{fig:stack-execution}). It is thus stronger than \Cref{amended-def} -- in fact, too strong.

The inclusion relations between the histories categorized as linearizable by the different versions of linearizability are illustrated in \Cref{fig:venn}.

\bibliographystyle{ACM-Reference-Format}
\bibliography{refs}

\end{document}

%% file: alternative-def.tex
\section{An Alternative Interpretation}
\label{alternative-def-section}\label{sec:alternative}

While the original paper only considers completed operations, i.e., an operation is a pair of an invocation and the next matching response, it might seem that if we also consider pending invocations as valid operations, as some papers do, then the original definition may work adequately, leading to an easy fix for the typo in the original definition. In this section we show that this is not the case.  Formally, an alternative definition for an operation would be:
\begin{definition}\label{alternative-op-def}
\emph{(Operation - alternative definition)}
An \emph{operation} in a history is either a pair consisting of an invocation and the next matching response, or a pending invocation only in case of an invocation that has no matching response in the execution. 
\end{definition}

This leads to an alternative interpretation of the linearizability definition with the typo, with $<_H$ applied to operations by \Cref{alternative-op-def} -- including pending invocations. This interpretation of operations does solve the problem pointed out in Section~\ref{section:issues} for the original linearizability definition, since it makes Condition L2 cover pending invocations. But this brings about a new problem, as not only linearized pending invocations are covered, but also pending invocations that are not linearized, namely, eliminated from $complete(H')$ and $S$ (which are equivalent by Condition L1). The definition by this interpretation might exclude an execution that seems legitimately linearizable: for an execution containing a pending invocation $e_2$ that cannot be legally linearized (and hence cannot appear in a linearization $S$) and an operation $e_1$ that precedes $e_2$, there exists no appropriate linearization $S$, since L2 implies that $e_1 <_S e_2$ and in particular that $e_2$ appears in $S$ (including a response, as $S$ is equivalent to $complete(H')$ by L1, which means it contains no pending invocations). 

An example of such a history that seems naturally linearizable but would be ruled out by the original definition with the alternative interpretation of an operation follows.
Consider a stack object, where if the stack is empty, then a popping process spins until an item is pushed into the stack.
Consider the history $H_s$ with processes $A$ and $B$ illustrated in \Cref{fig:stack-execution}.
The first pop’s response precedes the second pop’s invocation in $H_s$, hence $1^{st}\ Pop <_{H_s} 2^{nd}\ Pop$ (interpreting the happens-before relation relying on the alternative operation definition). Condition L2 of the linearizability definition implies that the same relation appears in $S$: $1^{st}\ Pop <_S 2^{nd}\ Pop$, and in particular the second pop appears in $S$.
However, the second pop cannot be included in $S$ because it cannot be legally completed on an empty stack. 

This problem regarding the definition of linearizability in the alternative interpretation is different from the problem with the original interpretation (pointed out in Section~\ref{section:issues}), but we remark that applying our fix, i.e., changing L2 to 
$<_{complete(H')}\subseteq<_S$, solves this problem as well and can make the alternative interpretation be an adequate (equivalent) definition for linearizability.

\setlength{\unitlength}{0.1mm}
\begin{figure}[t]
\caption{$H_s$, an execution on a stack with a second pop that cannot be completed}
\label{fig:stack-execution}
\begin{picture}(770,160)
\put(-100,45){$H_s$:}
\put(20,90){A}
\put(20,0){B}

\put(100,115){\line(0,-1){30}} 
\put(100,100){\line(1,0){200}} 
\put(300,115){\line(0,-1){30}} 
\put(143,115){Push(1)}

\put(350,25){\line(0,-1){30}} 
\put(350,10){\line(1,0){200}} 
\put(410,25){Pop()}
\put(550,25){\line(0,-1){30}} 
\put(542,40){1} 

\put(600,25){\line(0,-1){30}} 
\put(600,10){\line(1,0){200}} 
\put(660,25){Pop()}
\multiput(820,10)(20,0){3}{\line(1,0){5}}

\end{picture}

\end{figure}

%% file: equivalent-def.tex
\section{An Equivalent Definition} \label{equivalent-def-section}
According to the intuitive discussion in the original paper~\cite{herlihy1990linearizability}, linearizability provides the illusion that each operation takes effect instantaneously at some point between its invocation and its response. This point was later denoted a \emph{linearization point}. 
Referring to the intuitive meaning of linearizability, 
it makes sense that if a pending invocation takes effect, it does so instantaneously at some point after the invocation and before the end of the execution. 
Such an interpretation of linearizability has appeared in~\cite{lynch1996distributed}. Moreover, many data structure implementations \cite[e.g.][]{michael1996simple,michael2002high} are proven to be linearizable by listing linearization points as locations in the code for each of their methods. These locations naturally occur during the method call, after the invocation (and before the response or the end of the execution), obliviously of whether the invocation has a matching response in the original execution.

We next specify an equivalent definition of linearizability, similar to the \emph{atomicity} definition in \cite{lynch1996distributed}, which formalizes the above intuitive interpretation of linearizability.
\begin{definition}\label{equivalent-def}
\emph{(Linearizability by Linearization Points)}
A well-formed history $H$ is \emph{linearizable} if there exist distinct points in $H$, denoted linearization points, satisfying the following:
\begin{enumerate}
    \item
    For each operation, there exists a linearization point between its invocation and its response.
    \item
    There exists a subset $T$ of $H$'s pending invocations, such that for each invocation $inv$ in $T$ there exists a linearization point after the invocation, and there exists a response denoted $resp_{inv}$ for the invocation.
\end{enumerate}
Such that if we place each invocation that has a matching response and its matching response one right after another at their respective linearization point, do the same for each invocation $inv$ in $T$ and its response $resp_{inv}$, and exclude pending invocations not in $T$,
then the resulting sequence of invocations and responses, denoted $S$, is a legal sequential history.
\end{definition}

The original linearizability definition with the typo is not equivalent to \Cref{equivalent-def}. As a counterexample, $H_1$ demonstrated in \Cref{fig:register-histories} is linearizable by \Cref{orig-def} as shown in \Cref{examples-subsection}, but not by \Cref{equivalent-def}, since the pending write must be linearized for the read to return 1, but placing the write's linearization point after the read's one cannot yield a legal register sequential history.

We next show that fixing the typo in the linearizability definition makes the definition equivalent to \Cref{equivalent-def}, which formalizes the intuition behind linearizability.

\begin{claim}
\Cref{equivalent-def} (linearizability by linearization points) is equivalent to \Cref{amended-def} (amended linearizability).
\end{claim}
\begin{proof}
First we prove that \Cref{equivalent-def} implies \Cref{amended-def}:
Assume $H$ is linearizable by \Cref{equivalent-def}. Let $T$ and $S$ be a subset and a history that satisfy the definition. We will prove that $H$ is linearizable by \Cref{amended-def} with the same $S$.
Form $H'$ from $H$ by appending (in some arbitrary order) for each invocation in $T$, the response appended for it in $S$. 
L1 holds: For each process, $complete(H')$ is made of the same events as $S$. Their order is the same in $complete(H')$ and $S$, since $S$ is constructed by "shrinking" each operation (invocation and response) to its linearization point, which is placed by \Cref{equivalent-def} between an invocation and a following response.
L2 holds as well, since
if one operation precedes another operation in $complete(H')$, meaning the first operation's response happens before the second operation's invocation in $complete(H')$, then this order is preserved in $S$, in which the response of the first operation is moved to an earlier point (to the linearization point of the first operation) and the invocation of the second operation is moved to a later point (to the linearization point of the second operation).

We proceed to prove the other direction - \Cref{amended-def} implies  \Cref{equivalent-def}.
Assume $H$ is linearizable by \Cref{amended-def}. Let $H'$ and $S$ be an extension of $H$ and a linearization that satisfy the definition. 
We define $T$ to be the set of all pending invocations in $H$ for which a response is added in $H'$.
Next, we pick linearization points for operations of $H$ and for pending invocations of $T$, namely for operations of $complete(H')$, which are -- due to L1 -- the operations of $S$.
Denote the $i^{th}$ operation in $S$ by $e_i$.
We pick the linearization point of $e_i$ to be right after the later of the following: $e_i$'s invocation, and the linearization point of $e_{i-1}$. 

We will prove that our selected linearization points satisfy the requirements of \Cref{equivalent-def}.
First, Conditions (1) and (2) of \Cref{equivalent-def} hold as the linearization point of each $e_i$ is after its invocation in $complete(H')$ by definition, and also before its response as we next show. If not, it implies that the linearization point of $e_i$ is set right after the linearization point of $e_{i-1}$, and that the linearization point of $e_{i-1}$ is after $e_i$'s response. If the linearization point of $e_{i-1}$ was picked to be right after its invocation, it means that $e_i<_{complete(H')} e_{i-1}$ and we reach a contradiction by Condition L2 of \Cref{amended-def} (thanks to the typo fix). Else, it was picked to be right after the linearization point of $e_{i-2}$, and we continue with the same arguments until reaching a contradiction (the process is guaranteed to stop at $e_1$ at the latest).
Second, the linearization point order preserves the order of operations in $S$ by our definition of the linearization points. Therefore, the history constructed in \Cref{equivalent-def} by "shrinking" each operation of $S$ (namely, moving its invocation and response) to its linearization point, is equal to $S$ and is thus a legal sequential history.

\end{proof}